\documentclass[11pt]{article}
\usepackage[margin=1in]{geometry}
\usepackage{authblk}
\usepackage[utf8]{inputenc}

\usepackage{amsmath, amsthm, hyperref}
\usepackage{amsfonts}
\usepackage{amssymb}

\newtheorem{theorem}{Theorem}

\newtheorem{corr}{Corollary}

\newtheorem{defn}{Definition}

\newtheorem{lemma}{Lemma}
\newtheorem{fact}{Fact}

\DeclareMathOperator{\samp}{sc}
\DeclareMathOperator{\scz}{\samp_{\varepsilon, \delta}}
\DeclareMathOperator{\err}{err}
\DeclareMathOperator{\diam}{diam}
\DeclareMathOperator{\conv}{conv}
\DeclareMathOperator{\sign}{sign}
\DeclareMathOperator{\vol}{vol}

\newcommand{\junk}[1]{}

\newcommand{\epz}{\varepsilon}

\newcommand{\alg}{\mathcal{A}}
\newcommand{\uni}{\mathcal{U}}
\newcommand{\db}{\mathcal{D}}
\newcommand{\avgdb}{\overline{\db}}
\newcommand{\row}{e}
\newcommand{\quer}{\mathcal{Q}}

\renewcommand{\Pr}{\mathbb{P}}
\newcommand{\E}{\mathbb{E}}

\newcommand{\maps}{\colon}
\newcommand{\R}{\mathbb{R}}

\newcommand{\covn}{\mathcal{N}}

\newcounter{note}[section]
\renewcommand{\thenote}{\thesection.\arabic{note}}

\newcommand{\akis}[1]{\refstepcounter{note}$\ll${\bf Akis~\thenote:} {\sf #1}$\gg$\marginpar{\tiny\bf AK~\thenote}}

\title{Lower Bounds for Differential Privacy from Gaussian Width}
\author[1]{Assimakis Kattis \thanks{\tt{Email}:kattis@cs.toronto.edu}}
\author[2]{Aleksandar Nikolov \thanks{\tt{Email}:anikolov@cs.toronto.edu}}
\affil[1,2]{Department of Computer Science, University of Toronto.}
\date{}

\begin{document}
\maketitle

\begin{abstract}
 We study the optimal sample complexity of a given workload of linear queries under the constraints of differential privacy. The sample complexity of a query answering mechanism under error parameter $\alpha$ is the smallest $n$ such that the mechanism answers the workload with error at most $\alpha$ on any database of size $n$. Following a line of research started by Hardt and Talwar [STOC 2010],  we analyze sample complexity using the tools of asymptotic convex geometry. We study the sensitivity polytope, a natural convex body associated with a query workload that quantifies how query answers can change between neighboring databases. This is the information that, roughly speaking, is protected by a differentially private algorithm, and, for this reason, we expect that a ``bigger'' sensitivity polytope implies larger sample complexity. Our results identify the \emph{mean Gaussian width} as an appropriate measure of the size of the polytope, and show sample complexity lower bounds in terms of this quantity. Our lower bounds completely characterize the workloads for which the Gaussian noise mechanism is optimal up to constants as those having asymptotically maximal Gaussian width. 

Our techniques also yield an alternative proof of Pisier's Volume Number Theorem which also suggests an approach to improving the parameters of the theorem. 
\end{abstract}

\section{Introduction}

The main goal of private data analysis is to estimate aggregate
statistics while preserving individual privacy
guarantees. Intuitively, we expect that, for statistics that do not
depend too strongly on any particular individual, a sufficiently large
database allows computing an estimate that is both accurate and
private. A natural question then is to characterize the \emph{sample
  complexity} under privacy constraints: the smallest database size
for which we can privately estimate the answers to a given collection of queries
within some allowable error tolerance. Moreover, it is desirable to identify
algorithms that are simple, efficient, and have close to the best
possible sample complexity. In this work, we study these questions for
collections of \emph{linear queries} under the constraints of
\emph{approximate differential privacy}.

We model a \emph{database} $\db$ of \emph{size} $n$ as a multiset of
$n$ elements (counted with repetition) from an arbitrary finite
universe $\uni$. Each element of the database corresponds to the data
of a single individual. To define a privacy-preserving computation on
$\db$, we use the strong notion of \emph{differential privacy}.
Informally, an algorithm is {differentially private} if it has almost
identical behavior on any two databases $\db$ and $\db'$ that differ
in the data of a single individual. To capture this concept formally,
let us define two databases to be \emph{neighboring} if they have
symmetric difference of size at most $1$ (counted with
multiplicity). Then differential privacy is defined as follows:

\begin{defn}[\cite{DMNS}]
  A randomized algorithm $\alg$ that takes as input a database and
  outputs a random element from the set $Y$ satisfies \emph{$(\epz,
  \delta)$-differential privacy} if for all neighboring databases $\db,
  \db^{'}$ and all measurable $S \subseteq Y$ we have that: 
  \[
  \Pr[\alg(\db) \in S] \leq e^\epz
  \Pr[\alg(\db^{'}) \in S] + \delta,
  \]
  where probabilities are taken with respect to the randomness of $\alg$.
\end{defn}

One of the most basic primitives in private data analysis, and data
analysis in general, are \emph{counting queries} and, slightly more
generally, \emph{linear queries}. While interesting and natural
in themselves, they are also quite powerful: any statistical query (SQ)
learning algorithm can be implemented using noisy counting queries as a
black box~\cite{Kearns-SQ}. In our setting, we specify a linear query by a function
$q\maps\uni \to [0,1]$ (given by its truth table).  Slightly abusing
notation, we define the value of the query as $q(\db) =
\frac{1}{n}\sum_{\row \in \db}{q(\row)}$, where the elements of $\db$
are counted with multiplicity. For example, when $q\maps\uni\to
\{0,1\}$, we can think of $q$ as a property defined on $\uni$ and
$q(\db)$ as the fraction of elements of $\db$ that satisfy the
property: this is a \emph{counting query}. We call a set $\quer$ of
linear queries a \emph{workload} and an algorithm that answers a query
workload a \emph{mechanism}. We denote by $\quer(\db) = (q(\db))_{q
  \in \quer}$ the vector of answers to the queries in
$\quer$. Throughout the paper, we will use the letter $m$ for the size
of a workload $\quer$.

Starting from the work of Dinur and Nissim~\cite{DN}, it is known that
we cannot hope to answer too many linear queries too accurately while
preserving even a very weak notion of privacy. For this
reason, we must allow our private mechanisms to make some
error.  We
focus on \emph{average error} (in an $L_2$ sense). We define the
average error of an algorithm $\alg$ on a query workload $\quer$ and
databases of size at most $n$ as:
\begin{align*}
\err(\quer, \alg, n) &= \max_\db \left(\E \sum_{q \in
    \quer}{\frac{(\alg(\db)_q - q(D))^2}{|\quer|} }\right)^{1/2}
= \max_\db \left( \E\frac{1}{m}\|\alg(\db) - \quer(\db)\|_2^2\right)^{1/2},
\end{align*}
where the maximum is over all databases $\db$ of size at most $n$,
$\alg(\db)_q$ is the answer to query $q$ given by the algorithm $\alg$
on input $\db$, and expectations are taken with respect to the random
choices of $\alg$. This is a natural notion of error that also works particularly well with the geometric tools that we use. \junk{An
interesting direction for future research is to adapt our results for
worst-case error.}

In this work we study \emph{sample complexity}: the smallest database
size which allows us to answer a given query workload with error at
most $\alpha$.  The sample complexity of an algorithm $\alg$ with
error $\alpha$ is defined as: 
\[
\samp(\quer, \alg, \alpha) = \min \{ n: \err(\quer, \alg, n) \leq \alpha \}.
\] 
The sample complexity of answering the linear queries $\quer$ with
error $\alpha$ under $(\epz, \delta)$-differential privacy is defined
by: 
\[
\scz(\quer, \alpha) = \inf \{\samp(\quer, \alg, \alpha) : \alg
\mbox{ is } (\epsilon, \delta)\mbox{-differentially private} \}.
\]
The two main questions we are interested in are:
\begin{enumerate}
\item Can we characterize $\scz(\quer, \alpha)$ in terms of a natural
  property of the workload $\quer$?
\item Can we identify conditions under which simple and efficient $(\epz,
  \delta)$-differentially private mechanisms have nearly optimal
  sample complexity?
\end{enumerate}
We make progress on both questions. We identify a geometrically
defined property of the workload that gives lower bounds on the sample
complexity. The lower bounds also \emph{characterize} when
one of the simplest differentially private mechanisms, the Gaussian
noise mechanism, has nearly optimal sample complexity in the regime of
constant $\alpha$.

Before we can state our results, we need to
define a natural geometric object associated with a workload
of linear queries. This object has been important in applying
geometric techniques to differential
privacy~\cite{10vollb,12vollb,NTZ,smalldb}. 
\begin{defn}
  The \emph{sensitivity polytope} $K$ of a workload $\quer$ of $m$
  linear queries is equal to $K = \conv\{\pm \quer(\db): \db \text{ is
    a database of size } 1\}$.
\end{defn}
From the above definition, we see that $K$ is a symmetric (i.e. $K = -K$) convex polytope in $\R^m$. The
importance of $K$ lies in the fact that it captures how query answers
can change between neighboring databases: for any two neighboring
databases $\db$ and $\db'$ of size $n$ and $n'$ respectively,
$n\quer(\db) - n'\quer(\db') \in K$. This is exactly the information
that a differentially private algorithm is supposed to
hide. Intuitively, we expect that the larger $K$ is, the larger
$\scz(K,\alpha)$ should be. 

We give evidence for the above intuition, and propose the width of $K$
in a random direction as a measure of its ``size''. Let $h_K$ be the
support function of $K$: $h_K(y) = \max_{x \in K}{\langle x,
  y\rangle}$. For a unit vector $y$, $h_K(y) + h_K(-y)$ is the width
of $K$ in the direction of $y$; for arbitrary $y$, $h_K(ty)$ scales
linearly with $t$ (and is, in fact, a norm). We define the
$\ell^*$-norm of $K$, also known as its \emph{Gaussian mean width}, as
$\ell^*(K) = \E{[h_K(g)]}$, where $g$ is a standard Gaussian random
vector in $\R^m$. The following theorem captures our main result.
\begin{theorem}\label{thm:main}
  Let $\quer$ be a workload of $m$ linear queries, and let $K$ be its
  sensitivity polytope. The following holds for all  $\epz
  = O(1)$, $2^{-\Omega(n)} \leq \delta \leq 1/n^{1 +\Omega(1)}$, and
  any $\alpha \leq \frac{\ell^*(K)}{Cm(\log 2m)^2}$, where $C$
  is an absolute constant, and $\sigma(\epz,
\delta) = (0.5\sqrt{\epz} + \sqrt{2\log{(1/\delta)}})/\epz$:
  \begin{align*}
    \scz(\quer,\alpha) &= O\left(\min\Bigl\{\frac{\sigma(\epz,
        \delta)\ell^*(K)}{\sqrt{m}\alpha^2}, \frac{\sigma(\epz,\delta)\sqrt{m}}{\alpha}\Bigr\}\right);\\
    \scz(\quer,\alpha) &= \Omega\left(\frac{\sigma(\epz,\delta)\ell^*(K)^2}{m^{3/2}(\log  2m)^4\alpha} \right).
  \end{align*}
  The upper bound on sample complexity is achieved by a
  mechanism running in time polynomial in $m$, $n$, and
  $|\uni|$. Moreover, if $\ell^*(K) = \Omega(m)$, then $\scz(\quer,
  \alpha) =
  \Theta\Bigl(\frac{\sigma(\epz,\delta)\sqrt{m}}{\alpha}\Bigr)$ for
  any $\alpha \le 1/C$, where $C$ is an absolute constant. 
\end{theorem}
The sample complexity upper bounds in the theorem above are known from
prior work: one is given by the projection mechanism from~\cite{NTZ},
with the sample complexity upper bound in terms of $\ell^*(K)$ shown
in \cite{conjunctions}; the other upper bound is given by the Gaussian
noise mechanism~\cite{DN,DworkN04,DMNS}. The main new contribution in
this work are the lower bounds on sample complexity. The gap between
upper and lower bounds is small when $\ell^*(K)$ is close to its
maximal value of $m$. Indeed, when $\ell^*(K) = \Theta(m)$, our
results imply that the Gaussian noise mechanism has optimal sample
complexity up to constants. This is, to the best of our knowledge, the
first example of a general geometric condition under which a simple
and efficient mechanism has optimal sample complexity up to
\emph{constant} factors. Moreover, in the constant error regime this
condition is also \emph{necessary} for the Gaussian mechanism to be optimal
up to constants: when $\ell^*(K) = o(m)$ and $\alpha = \Omega(1)$, the
projection mechanism has asymptotically smaller sample complexity than
the Gaussian mechanism.

We can prove somewhat stronger results for another natural problem in
private data analysis, which we call the \emph{mean point problem}. In
this problem, we are given a closed convex set $K\subset \R^m$, and we
are asked to approximate the mean $\avgdb$ of the database $\db$,
where $\db = \{x_1, \ldots, x_n\}$ is a multiset of points in $K$ and
$\overline{\db} = \frac{1}{n}\sum_{i=1}^n{x_i}$.  This problem, which
will be the focus for most of this paper, has a more geometric flavor,
and is closely related to the query release problem for linear
queries. In fact, Theorem~\ref{thm:main} will essentially follow from a
reduction from the results below for the mean point problem.

With respect to the mean point problem, we define the error of an
algorithm $\alg$ as:
\[
\err(K,\alg,n) = \sup_{\db}(\E\|\alg(\db) - \avgdb\|_2^2)^{1/2},
\]
where the supremum is over databases $\db$ consisting of at most $n$
points from $K$, and the expectation is over the randomness of the
algorithm. The sample complexity of an algorithm $\alg$ with error
$\alpha$ is defined as: \[\samp(K, \alg, \alpha) = \min \{ n : \err(K,
\alg, n) \leq \alpha \}.\] The sample complexity of solving the mean
point problem with error $\alpha$ over  $K$ is defined
by: \[\scz(K, \alpha) = \min \{\samp(K, \alg, \alpha) : \alg \mbox{ is
} (\epsilon, \delta)\mbox{-differentially private} \}.\]

Our main result for the mean point problem is given in the following
theorem:
\begin{theorem}\label{thm:meanpt-main}
  Let $K$ be a symmetric convex body contained in the unit Euclidean ball
  $B_2^m$ in $\R^m$. The following holds for all  $\epz
  = O(1)$, $2^{-\Omega(n)} \leq \delta \leq 1/n^{1 +\Omega(1)}$, and
  any $\alpha \leq \frac{\ell^*(K)}{C\sqrt{m}(\log 2m)^2}$, where $C$
  is an absolute constant, and $\sigma(\epz,
\delta) = (0.5\sqrt{\epz} + \sqrt{2\log{(1/\delta)}})/\epz$:
  \begin{align*}
    \scz(K,\alpha) &= O\left(\min\Bigl\{\frac{\sigma(\epz,
        \delta)\ell^*(K)}{\alpha^2},
      \frac{\sigma(\epz,\delta)\sqrt{m}}{\alpha}\Bigr\}\right);\\
    \scz(K,\alpha) &=
    \Omega\left(\frac{\sigma(\epz,\delta)\ell^*(K)}{(\log
        2m)^2\alpha}\right).
  \end{align*}
  The upper bound on sample complexity is achieved by a mechanism
  running in time polynomial in $m$, $n$, and $|\uni|$. Moreover, when
  $\ell^*(K) = \Omega(\sqrt{m})$, then $\scz(\quer, \alpha) =
  \Theta\Bigl(\frac{\sigma(\epz,\delta)\sqrt{m}}{\alpha}\Bigr)$ for
  any $\alpha \le 1/C$, where $C$ is an absolute constant.
\end{theorem}
The upper bounds again follow from prior work, and in fact are also
given by the projection mechanism and the Gaussian noise mechanism,
which can be defined for the mean point problem as well.  Notice that
the gap between the upper and the lower bound is on the order
$\frac{(\log 2m)^2}{\alpha}$. If the lower bound was valid for all
values of the error parameter $\alpha$ less than a fixed constant,
rather than for $\alpha \le  \frac{\ell^*(K)}{C\sqrt{m}(\log
    2m)^2}$, Theorem~\ref{thm:meanpt-main} would nearly
characterize the optimal sample complexity for the mean point problem
for all constant $\alpha$. Unfortunately, the restriction on $\alpha$
is, in general, necessary (up to the logarithmic terms) for lower
bounds on sample complexity in terms of $\ell^*(K)$. For example, we
can take $K = \gamma B_2^m$, i.e.~a Euclidean ball in $\R^m$ with
radius $\gamma$. Then, $\ell^*(K) = \Theta(\gamma\sqrt{m})$, but the
sample complexity is $0$ when $\alpha > \gamma$, since the trivial
algorithm which ignores the database and outputs $0$ achieves error
$\gamma$. Thus, a more sensitive measure of the size of $K$ is
necessary to prove optimal lower bounds. We do, nevertheless, trust
that the techniques introduced in this paper bring us closer to this
goal.  

\junk{Moreover, the projection mechanism does not have optimal
sample complexity for every $K$. For example, if we take the cone $K =
\{x: \sqrt{1 - \alpha^2}\|x\|_2 \le x_1 \le \sqrt{1 - \alpha^2}\}$,
then the projection mechanism has sample complexity
$\Omega(\frac{\sqrt{m}}{\alpha^2})$ for error $2\alpha$, but
projecting on the line segment joining $0$ and $\sqrt{1 -
  \alpha^2}e_1$ has sample complexity $O(\frac{1}{\alpha^2})$. }

We conclude this section with a high-level overview of our
techniques. Our starting point is a recent tight lower bound on the
sample complexity of a special class of linear queries: the $1$-way
marginal queries. These queries achieve the worst case sample
complexity for a family of $m$ linear queries:
$\Omega(\sqrt{m}/\alpha)$~\cite{BunUV13,SteinkeU15}. The sensitivity
polytope of the $1$-way marginals is the cube $[-1, 1]^m$, and it can
be shown that the lower bound on the sample complexity of $1$-way
marginals implies an analogous lower bound on the sample complexity of
the mean point problem with  $K = Q^m = [-1/\sqrt{m},
1/\sqrt{m}]^m$. For the mean point problem, it is easy to see that when
$K' \subseteq K$, the sample complexity for $K'$ is no larger than the
sample complexity for $K$. Moreover, we can show that the sample
complexity of any projection of $K$ is no bigger than the sample
complexity of $K$ itself. So, our strategy then is to find a large scaled
copy of $Q^{m'}$, $m' \le m$, inside a projection of $K$ onto a large dimensional
subspace whenever $\ell^*(K)$ is large. We solve this geometric
problem using deep results from asymptotic convex geometry, namely the
Dvoretzky criterion, the low $M^*$ estimate, and the $MM^*$
estimate. 

Our techniques also yield an alternative proof of the volume
number theorem of Milman and Pisier~\cite{MilmanPisier87}. Besides
avoiding the quotient of subspace theorem, our proof yields an
improvement in the volume number theorem, conditional on the well-known
conjecture that any symmetric convex body $K$ has a position (affine
image) $TK$ for which $\ell^*(TK)\ell(TK) = O(m\sqrt{\log 2m})$, where
$\ell(K)$ is the expected $K$-norm of a standard Gaussian. More
details about this connection are given in Section~\ref{sect:arbirary}.

\subsection{Prior Work}

Most closely related to our work are the results of Nikolov, Talwar,
and Zhang~\cite{NTZ}, who gave a private mechanism (also based on the
projection mechanism, but more involved) which has nearly optimal
sample complexity (with respect to average error), up to factors
polynomial in $\log m$ and $\log |\uni|$. This result was subsequently
improved by Nikolov~\cite{smalldb}, who showed that the $\log m$
factors can be replaced by $\log n$. While these results are nearly
optimal for subconstant values of the error parameter $\alpha$, i.e.~the
optimality guarantees do not depend on $1/\alpha$, factors polynomial
in $\log |\uni|$ can be prohibitively large. Indeed, in many natural
settings, such as that of marginal queries, $|\uni|$ is exponential in
the number of queries $m$, so the competitiveness ratio can be
polynomial in $m$. \junk{Here we remove any dependence on $|\uni|$ at the
cost of introducing a dependence on $1/\alpha$. However, the setting
of constant error is very natural and we consider this a 
reasonable trade-off.}

The line of work that applies techniques from convex geometry to
differential privacy started with the beautiful paper of Hardt and
Talwar~\cite{10vollb}, whose results were subsequently strengthened
in~\cite{12vollb}. These papers focused on the ``large database'' regime
(or, in our language, the setting of subconstant error), and pure
differential privacy ($\delta = 0$). \junk{In this paper, we focus instead
on constant error and approximate differential privacy ($\delta >
0$). }

\section{Preliminaries}

We begin with the introduction of some notation. Throughout the paper
we use $C$,  $C_1$, etc., for absolute constants, whose value may change
from line to line. We use $\|\cdot \|_2$ for the Euclidean norm, and
$\|\cdot\|_1$ for the $\ell_1$ norm. We define $B_1^m$ and $B^m_2$ to
be the $\ell_1$ and $\ell_2$ unit balls in $\mathbb{R}^m$
respectively, while $Q^m = [-\frac{1}{\sqrt{m}}, \frac{1}{\sqrt{m}}]^m
\subseteq \mathbb{R}^m$ will refer to the $m$-dimensional hypercube,
normalized to be contained in the unit Euclidean ball. We use $I_m$
for the identity operator on $\R^m$, as well as for the $m\times m$
identity matrix. For a given subspace $E$, we define $\Pi_E \maps
\mathbb{R}^m \to \mathbb{R}^m$ as the orthogonal projection operator
onto $E$. Moreover, when $T\maps E \to F$ is a linear operator between
the subspaces $E, F\subseteq \R^m$, we define $\| T \| =
\max\{\|Tx\|_2: \|x\|_2 = 1\}$ as its operator norm, which is also
equal to its largest singular value $\sigma_1(T)$.  For the diameter
of a set $K$ we use the nonstandard, but convenient, definition
$\diam{K} = \max{\{ \|x\|_2 : x \in K\}}$. For sets symmetric around
$0$, this is equivalent to the standard definition, but scaled up by a
factor of $2$. We use $N(\mu, \Sigma)$ to refer to the Gaussian
distribution with mean $\mu$ and covariance $\Sigma$, and we use the
notation $x \sim N(\mu, \Sigma)$ to denote that $x$ is distributed as
a Gaussian random variable with mean $\mu$ and covariance
$\Sigma$. For a $m\times m$ matrix (or equivalently an operator from
$\ell_2^m$ to $\ell_2^m$) $A$ we use $A \succeq 0$ to denote that $A$
is positive semidefinite. For positive semidefinite matrices/operators
$A$, $B$, we use the notation $A \preceq B$ to denote $B-A \succeq
0$.

\section{Probability Theory}
\label{app:prob}

We make use of some basic comparison theorems from the theory of
stochastic processes. First we state the well-known symmetrization
lemma. We also give the short proof for completeness. Recall that $\xi_1,
\ldots, \xi_n$ are a sequence of \emph{Rademacher random variables} if each
$\xi_i$ is uniformly and independently distributed in $\{-1, 1\}$.

\begin{lemma}[Symmetrization]\label{lm:sym}
  Let $p>1$, and let $\|\cdot \|$ be a norm on $\R^m$. Then, for any sequence $x_1, \ldots, x_n$ of
  independent random variables in $\R^m$ such that $\E\|x_i\|^p$ is
  finite for every $i$, we have
  \[
  \E \Big\|\sum_i{x_i} - \E\sum_i{x_i}\Big\|^p \le 2^p \E \Big\|\sum_i{\xi_i x_i}\Big\|^p,
  \]
  where $\xi_1, \ldots, \xi_n $ are Rademacher random variables,
  independent of $x_1, \ldots, x_n$. Each expectation above is with
  respect to all random variables involved.
\end{lemma}
\begin{proof}
  Let $x'_1, \ldots, x'_n$ be independent copies of $x_1, \ldots,
  x_n$. Then, $\E\sum_i{x_i} = \E\sum_i{x'_i}$, and, by convexity of
  the function $\|\cdot \|^p$ and Jensen's inequality,
  \[
    \E \Big\|\sum_i{x_i} - \E\sum_i{x_i}\Big\|^p =   \E \Big\|\sum_i{x_i} -
    \E\sum_i{x'_i}\Big\|^p \le   \E \Big\|\sum_i{x_i - x'_i}\Big\|^p.
  \]
  Because $x_i$ and $x'_i$ are independent and identically
  distributed, the random variables $x_i - x'_i$ and $x'_i - x_i$ are also
  identically distributed. Therefore, $\E \|\sum_i{x_i -
    x'_i}\|^p = \E \|\sum_i{\xi_i (x_i - x'_i)}\|^p$, where $\xi_1,
  \ldots, \xi_n$ are independent Rademacher random variables, as in
  the statement of the lemma. Finally, by Minkowski's inequality
  (i.e.~triangle inequality for $L_p$), we have
  \[
  \E \Big\|\sum_i{\xi_i (x_i - x'_i)}\Big\|^p \le \Bigg(\Big(\E \Big\|\sum_i{\xi_i
    x_i}\Big\|^p\Big)^{1/p} +  \Big(\E \|\sum_i{\xi_i x'_i}\|^p\Big)^{1/p} \Bigg)^p 
  = 2^p \E \Big\|\sum_i{\xi_i x_i}\Big\|^p,
  \]
  as desired.
\end{proof}

Next we state a simple comparison theorem for Gaussian random variables.
\begin{lemma}\label{lm:comp-gauss}
  Let $x \sim N(0,\Sigma)$ and $x' \sim N(0, \Sigma')$ be Gaussian
  random variables in $\R^m$, and assume $\Sigma
  \preceq \Sigma'$. Then, for any norm $\|\cdot\|$ on $\R^m$, we have:
  \[
  \E \|x\| \le \E \|x'\|.
  \]
\end{lemma}
\begin{proof}
  Couple $x$ and $x'$ so that
  they are independent, and define a new random variable $y \sim N(0,
  \Sigma' - \Sigma)$, independent of $x$ and $x'$. Then the random
  variables $x + y$ and $x - y$ are distributed identically to $x'$,
  and, by linearity of expectation and the triangle inequality we have
  \[
  \E\|x'\| = \E\left[\frac{\|x + y\| + \|x - y\|}{2}\right] \ge \E\|x\|.
  \]
  This completes the proof.
\end{proof}
Note that the same conclusion follows under weaker assumptions from
Slepian's lemma.

\subsection{Convex Geometry}

In this section, we outline the main geometric tools we use in later
sections. For a more detailed treatment, we refer to the lecture notes
by Vershynin~\cite{Ver-notes} and the books by
Pisier~\cite{Pisier-book} and Artstein-Avidan, Giannopoulos, and
Milman~\cite{AGM-book}.

Throughout, we define a \emph{convex body} $K$ as a compact subset of $\R^m$
with non-empty interior. A convex body $K$ is \emph{(centrally) symmetric} if and
only if $K = -K$.
We define  the \emph{polar body} ${K}^\circ$ of $K$ as: $K^\circ = \{ y
: \langle x, y \rangle \leq 1 \mbox{  } \forall x \in K \}.$ The
following basic facts are easy to verify and very useful.

\begin{fact}{\label{polarsubset}}
For convex bodies $K, L \subseteq \mathbb{R}^m$, $K \subseteq L \Leftrightarrow L^\circ \subseteq K^\circ.$
\end{fact}

\begin{fact}[Section/Projection Duality]{\label{polarduality}}
For a convex body $K \subseteq \mathbb{R}^m$ and a subspace $E \subseteq \mathbb{R}^m$:
\begin{enumerate}
\item $(K \cap E)^\circ = \Pi_E(K^\circ);$
\item $(\Pi_E(K))^\circ = K^\circ \cap E.$
\end{enumerate}
\noindent In both cases, the polar is taken in the subspace $E$.
\end{fact}

\begin{fact}\label{polarmap}
  For any invertible linear map $T$ and any convex body $K$,
  $T(K)^\circ = T^{-*}(K^\circ)$, where $T^{-*}$ is the inverse of the
  adjoint operator $T^*$.
\end{fact}

A simple special case of Fact~\ref{polarmap} is that, for any convex body $K$,
$(rK)^\circ = \frac1r K^\circ$. Using this property alongside
Fact~\ref{polarsubset}, we have the following useful corollary.

\begin{corr}{\label{dualequivalence}}
For a convex body $K \subseteq \mathbb{R}^m$ and $E \subseteq \mathbb{R}^m$ a subspace with $k = \dim{E}$, the following two statements are equivalent:

\begin{enumerate}
\item $\Pi_E(r B_2^m) \subseteq \Pi_E(K);$
\item $K^\circ \cap E \subseteq \frac{1}{r}(B_2^m \cap E),$
\end{enumerate}

\noindent where, as before, taking the polar set is considered in the
subspace $E$. Notice that the second statement is also equivalent to
$\diam(K^\circ \cap E) \le \frac1r$.
\end{corr}



Our work relies on appropriately quantifying the ``size'' of
(projections and sections of) a convex body. It turns out that, for our
purposes, the right measure of size is related to the notion of
\emph{width}, captured by the \emph{support function}. Recall from the
introduction that the support function of a convex body $K\subset
\R^m$ is given by $h_K(y) = \max_{x \in K}{\langle x, y\rangle}$ for
every $y \in \R^m$. The support function is intimately related to the
\emph{Minkowski norm} $\|\cdot \|_K$, defined for a symmetric convex
body $K\subseteq \R^m$ by  $ \| x \|_K = \min{\{ r \in \mathbb{R} : x
  \in rK\}},$ for every $x \in \R^m$. It is easy to verify that
$\|\cdot\|_K$ is indeed a norm. The support function $h_K$ is
identical to the Minkowski norm of the polar body $K^\circ$ (which is
also the dual norm to $\|\cdot \|_K$): $h_K(y) = \|y\|_{K^\circ}$ for every $y \in
\R^m$.

Now we come to the measure of the ``size'' of a convex body which will
be central to our results: the Gaussian mean width of the body,
defined next.
\begin{defn}
The Gaussian mean width and Gaussian mean norm of a symmetric convex
body  $K \subseteq \mathbb{R}^m$ are defined respectively as: 
\begin{align*}
\ell^*(K) &= \E\|g\|_{K^\circ} = \E[h_K(g)],
&\ell(K) = \E\|g\|_{K},
\end{align*}
where ${g \sim N(0,I_m)}$ is a standard Gaussian random variable.
\end{defn}

The next lemma gives an estimate of how the mean width changes when
applying a linear transformation to $K$.

\begin{lemma}\label{lm:comp-width}
  For any symmetric convex body
  $K\subset \R^m$, and any linear operator $T\maps \ell_2^m \to \ell_2^m$:
  \begin{align*}
    \ell^*(T(K)) \le \|T\|\ell^*(K).
  \end{align*}
\end{lemma}

\begin{proof}
  Notice that, for a standard Gaussian $g \sim N(0, I_m)$,
  \[
  \ell^*(T(K)) = \E \sup_{x \in K}{\langle T(x), g \rangle}
  = \E \sup_{x \in K}{\langle x, T^*(g) \rangle} = \E \|T^*(g)\|_{K^\circ}.
  \]
  Treating $T^*T$ as an $m\times m$ matrix in the natural way, we see
  that $T^*(g) \sim N(0,T^*T)$. By applying Lemma~\ref{lm:comp-gauss} to
  $\Sigma = T^*T$ and $\Sigma' = \|T\|^2 I_m$, we have that
  \[
  \ell^*(T(K)) = \E \|T^*(g)\|_{K^\circ} \le \|T\|\cdot\E \|g\|_{K^\circ} = \|T\|\ell^*(K).
  \]
  This finishes the proof of the lemma.
\end{proof}

Similar to approaches in previous works (\cite{10vollb}, \cite{NTZ}), we exploit properties inherent to a specific position of $K$ to prove lower bounds on its sample complexity. \junk{More specifically, we use the $\ell$-position of $K$, defined below, to show an upper bound on its mean width.}

\begin{defn}[$\ell$-position]
A convex body $K\subseteq \R^m$ is in $\ell$-position if for all
linear operators $T\maps \ell_2^m \to \ell_2^m$: 
\[\ell^*(K)\cdot \ell(K) \leq \ell^*(T(K)) \cdot \ell(T(K)).\]
\end{defn}

Clearly, $K$ is in $\ell$-position if and only if $K^\circ$ is in
$\ell$-position, since $\ell^*(K) = \ell(K^\circ)$ for any convex body
$K$. 
Note further that the product $\ell^*(K)\cdot \ell(K)$ is scale-invariant, in
the sense that $\ell^*(rK)\cdot \ell(rK) = \ell^*(K)\cdot \ell(K)$ for
any real $r$. This is because, for any $x, y \in \R^m$, $\|x\|_{rK} =
\frac1r \|x\|_K$, and 
$h_{rK}(y) = r h_K(y)$, so
$\ell^*(rK) = r\ell^*(K)$ and $\ell(rK) = \frac1r \ell(K)$. 

We will relate the Gaussian mean width of $K$ to another measure of its
size, and the size of its projections and sections, known as Gelfand
width. A definition follows.

\junk{
\begin{defn}[Covering Number]
For two sets $X,Y \subset \mathbb{R}^m$, the covering number $\covn(X,Y)$ is the minimum number of translations of $Y$ needed to completely cover $X$. 
Furthermore, we refer to $\log{\covn(X,Y)}$ as the \textit{metric entropy} of $X$ with respect to $Y$.
\end{defn}}

\begin{defn}[Gelfand width]
For two symmetric convex bodies $K,L \subset \mathbb{R}^m$, the \emph{Gelfand
width of order $k$} of $K$ with respect to $L$ is defined as: 
\[
c_k(K,L) = \inf_E \inf\{r: K \cap E \subseteq r(L\cap E)\} = \inf_E
\sup\{\|x\|_L: x \in K \cap E\},
\]
where the first infimum is over subspaces $E \subseteq \R^m$ of
co-dimension at most $k-1$ (i.e.~of dimension at least $m-k+1$). When
$k>m$, we define $c_k(K,L) = 0$.
We denote $c_k(K) = c_k(K,B_2^m)$, and we call $c_k(K)$ simply the
Gelfand width of $K$ of order $k$.
\end{defn}

Note that $c_k(K) = \inf_E \diam(K \cap E)$, where the infimum is over
subspaces $E \subseteq \R^m$ of codimension at most $k-1$. Observe
also that for any $K$ and $L$, $c_k(K,L)$ is non-increasing in $k$. It
is well-known that the infimum in the definition is actually achieved~\cite{Pinkus-widths}.

\subsection{Composition of Differential Privacy}

\junk{In this section we introduce basic and facts from the
theory of differential privacy. For more details, proofs, and
references we refer to the book of Dwork and Roth~\cite{DP-book}.

A \emph{database} $\db$ of \emph{size} $n$ drawn from a
\emph{universe} $\uni$ is a multiset of $n$ elements (counted with
repetition) from $\uni$. Two databases are \emph{neighboring} if they
have symmetric difference of size at most $1$.
\akis{Should we define symmetric difference?}
We have the following basic definition.
\begin{defn}
  A randomized algorithm $\alg$ that takes as input a database and
  outputs a random element from the set $Y$ satisfies \emph{$(\epz,
  \delta)$-differential privacy} if for all neighboring databases $\db,
  \db^{'}$ and all measurable $S \subseteq Y$ we have that: 
  \[
  \Pr[\alg(\db) \in S] \leq e^\epz
  \Pr[\alg(\db^{'}) \in S] + \delta,
  \]
  where probabilities are taken with respect to the randomness of $\alg$.
\end{defn}} 

One of the most important properties of differential privacy is that it behaves
nicely under (adaptively) composing mechanisms.
\begin{lemma}[Composition]\label{lm:composition}
For randomized algorithms $\alg_1$ and  $\alg_2$ satisfying $(\epz_1, \delta_1)$- and $(\epz_2, \delta_2)$-differential privacy respectively, the algorithm $\alg(\db) = (\alg_1(\db), \alg_2(\alg_1(\db),\db))$ satisfies $(\epz_1 + \epz_2, \delta_1 + \delta_2)$-differential privacy.
\end{lemma}

\subsection{Known Bounds}

In this section, we recall some known differentially private
mechanisms, with bounds on their sample complexity, as well as a lower
bound on the optimal sample complexity. We start with the lower bound:

\begin{theorem}[\cite{BunUV13,SteinkeU15}]\label{thm:tightcube}
For all $\epz = O(1)$, $2^{-\Omega(n)} \leq \delta \leq 1/n^{1 +
  \Omega{(1})}$ and $\alpha \leq 1/10$:

\begin{equation}
\scz{(Q^m,\alpha)} = \Omega{\left( \frac{\sqrt{m \log{1/\delta}}}{\alpha \epz} \right)}.
\end{equation}

\end{theorem}

Next we recall one of the most basic mechanisms in differential
privacy, the Gaussian mechanism. A proof of the privacy guarantee,
with the constants given below, can be found in~\cite{NTZ}.
\begin{theorem}[Gaussian Mechanism~\cite{DN,DworkN04,DMNS}]\label{thm:gauss-mech}
Let $\db = \{x_1, \ldots, x_n\}$ be such that $\forall i: \|x_i\|_2 \leq
\sigma$. If $w \sim N(0, \sigma(\epz, \delta)^2\sigma^2I_m)$, $\sigma(\epz,
\delta) = (0.5\sqrt{\epz} + \sqrt{2\log{(1/\delta)}})/\epz$ and $I_m \in
\mathbb{R}^{m\times m}$ is the identity matrix, then the algorithm $\alg_{GM}$
defined by $ \alg_{GM}(\db) = \avgdb + \frac{1}{n}w$ is $(\epz,
\delta)$-differentially private.
\end{theorem}

\begin{corr}{\label{gaussupper}}
For any symmetric convex $K \subseteq B_2^m$: \[ \scz(K, \alpha) = O\left(\frac{\sqrt{m \log{1/\delta}}}{\alpha\epz} \right).\]
\end{corr}

In the rest of the paper we will use the notation $\sigma(\epz,
\delta) = \frac{0.5\sqrt{\epz} + \sqrt{2\log{(1/\delta)}}}{\epz}$ from the
theorem statement above. 

Finally, we also present the projection mechanism
from~\cite{NTZ}, which post-processes the output of the Gaussian
mechanism by projecting onto $K$.

\begin{theorem}[Projection Mechanism~\cite{NTZ,conjunctions}]\label{thm:general-proj}
  Let $K\subseteq B_2^m$ be a symmetric convex body, and
  define $\alg_{PM}$ to be the algorithm that, on input $\db
  = \{x_1, \ldots, x_n\} \subset K$, outputs:
  \[
  \hat{y} = \arg\min\{\|\hat{y} - \tilde{y}\|_2^2: \hat{y} \in K\}, 
  \]
  where $\tilde{y} = \avgdb + \frac{1}{n}w$, $w \sim N(0, \sigma(\epz,
  \delta)^2I_m)$. Then $\alg_{PM}$ satisfies $(\epz,
  \delta)$-differential privacy and has sample complexity:
  \[
  \samp(K, \alg_{PM}, \alpha) = O\left(\frac{\sigma(\epz, \delta)\ell^*(K)}{\alpha^2}\right).
  \]
  \junk{
  Moreover, there exists a constant $C > 0$ such that, for any $t >
  0$, we have
  \[
  \Pr\left[ \|\alg_{PM}(\db) - \avgdb\|_2^2 >
  (1+t)\frac{C\sigma(\epz,\delta)\ell^*(K)}{n}\right] \le e^{-t^2/C}.
  \]}
\end{theorem}

\begin{corr}{\label{projupper}}
For any symmetric convex $K \subseteq B_2^m$: \[\scz(K, \alpha) = O\left(\frac{\sigma(\epz, \delta)\ell^*(K)}{\alpha^2}\right).\]
\end{corr}

\section{Basic Properties of Sample Complexity \label{sec:bp}}

In this section, we prove some fundamental properties of sample
complexity that will be extensively used in later sections.

\begin{lemma}\label{lemma:subsetsc}
$L \subseteq K \Rightarrow \forall \alpha \in (0,1): \scz(L,\alpha) \leq \scz(K,\alpha)$.
\end{lemma}

\begin{proof}
  Observe first that for any algorithm $\alg$ and any $n$,
  $\err(L, \alg, n) \le \err(K, \alg, n)$, because $\err(K, \alg, n)$
  is a supremum over a larger set than $\err(L, \alg, n)$. This
  implies that $\samp(K,\alg, \alpha) \le \samp(L, \alg, \alpha)$ holds
  for any algorithm $\alg$, and, in particular, for the
  $(\epz,\delta)$-differentially private algorithm $\alg^*$ that achieves
  $\scz(K,\alpha)$. Then, we have:
  \[
  \scz(L, \alpha) \le \samp(L, \alg^*, \alpha) \le \samp(K,  \alg^*,
  \alpha) = \scz(K, \alpha),
  \]
  as desired.
\end{proof}

\begin{corr}\label{b2bound}
For all $\epz = O(1)$, $2^{-\Omega(n)} \leq \delta \leq 1/n^{1 +
  \Omega{(1})}$ and $\alpha \leq 1/10$: 
\[\scz(B_2^m, \alpha) = \Omega{\left( \frac{\sqrt{m \log{1/\delta}}}{\alpha \epz} \right)}.\]
\end{corr}
\begin{proof}
Since $Q^m \subseteq B_2^m$, this follows directly from Lemma~\ref{lemma:subsetsc} and Lemma~\ref{thm:tightcube}.
\end{proof}

\begin{lemma}\label{lemma:linop2}
  For any $\alpha \in (0,1)$, any linear operator $T \maps
  \R^m \rightarrow \R^m$ and any symmetric convex body
  $K \subset \mathbb{R}^m$: $$ \scz(K, \alpha) \geq \scz( T(K), \alpha
  \cdot \|T\|).$$
\end{lemma}

\begin{proof}
  Let $\alg$ be an $(\epz, \delta)$-differentially private algorithm
  that achieves $\scz(K, \alpha)$. Fix a function $f\maps T(K) \to K$
  so that for every $x' \in T(K)$, $f(x') \in T^{-1}(x') \cap K$. We define a
  new algorithm $ \alg'$ that takes as input $\db' = \{x'_1, ...,
  x'_n\} \subset T(K)$ and outputs $\alg'(\db') = T(\alg(\db))$, where
  $\db = f(\db') =
  \{f(x'_1), \ldots, f(x'_n)\}$.
  We claim that $\alg'$ is
  $(\epz, \delta)$-differentially private and that $\err(T(K),
  \alg', n) \le \|T\|\cdot \err(K, \alg, n)$ holds for
  every $n$. This claim is sufficient to prove the lemma, because it
  implies:
  \[\scz(T(K), \alpha\cdot\|T\|)
  \le \samp(T(K), \alg', \alpha\cdot \|T\|) \le
  \samp(K,\alg, \alpha) = \scz(K, \alpha).\]
  To show the claim, first observe that, by linearity,
  $\overline{\db'} = T(\avgdb)$. We get: 
  \begin{align*}
  \E \|\alg'(\db') - \overline{\db'}\|^2_2 &=
  \E \|T(\alg(\db)) - T(\avgdb) \|^2_2 \\
  &= \E\|T( \alg(\db) - \avgdb)\|^2_2 \\
  &\leq \| T \|^2 \cdot \E\|\alg(\db) - \avgdb\|^2_2 \le \| T \|^2 \err(K,\alg,n)^2,
  \end{align*}
  where the first inequality follows by the definition of the operator
  norm. Since this holds for arbitrary $n$ and $\db' \subset
  T(K)$ of size $n$, it implies the claim on the error bound of
  $\alg'$.  It remains to show that
  $\alg'$ is $(\epz, \delta)$-differentially
  private. Note that for every two neighboring databases
  $\db_1'$ and $\db_2'$ of points in $T(K)$,the
  corresponding databases $\db_1 = f(\db_1')$ and $\db_2 = f(\db_2')$ of points in $K$ are also
  neighboring. Then $\alg(\db) = \alg(f(\db'))$ is $(\epz, \delta)$-differentially
  private as a function of $\db'$, and the privacy of
  $\alg'$ follows from Lemma \ref{lm:composition}.
\end{proof}

\begin{corr}\label{scaling}
  For any $t > 0$: \[ \scz(tK, t\alpha) = \scz(K, \alpha).\]
\end{corr}
\begin{proof}
  Taking $T = tI_m$ in Lemma \ref{lemma:linop2}, where $I_m$ is the
  identity on $\R^m$, the lemma implies $\scz(tK, t\alpha) \le \scz(K,
  \alpha)$. Since this inequality holds for any $t$ and $K$, we may
  apply it to $K' = tK$ and $t' = 1/t$, and we get $\scz(K, \alpha) =
  \scz((1/t)tK, (1/t)t\alpha) \le \scz(tK, t\alpha)$. 
\end{proof}

Since for any subspace $E$ of $\R^m$, the corresponding orthogonal
projection $\Pi_E$ has operator norm $1$, we also immediately get the
following corollary of Lemma \ref{lemma:linop2}:
\begin{corr}{\label{orthproj}}
For any subspace $E$: \[\scz(K, \alpha) \geq \scz( \Pi_E(K), \alpha).\]
\end{corr}

In the next theorem, we combine the lower bound in \autoref{b2bound} and the
properties we proved above in order to give a lower bound on the
sample complexity of an arbitrary symmetric convex body $K$ in terms
of its geometric properties. In the following sections we will relate
this geometric lower bound to the mean Gaussian width of $K$.

\begin{theorem}[Geometric Lower Bound]{\label{modularlb}}
  For all $\epz = O(1)$, $2^{-\Omega(n)} \leq \delta \leq 1/n^{1 +
    \Omega{(1})}$, any convex symmetric body $K \subseteq
  \mathbb{R}^m$, any $1 \le k \le m$ and any $\alpha \le
  1/(10c_{k}(K^\circ))$: 
  
  \[
  \scz(K, \alpha) 
  = \Omega \left(
    \frac{\sqrt{\log{1/\delta}}}{\alpha \epz} 
    \cdot \frac{\sqrt{m - k + 1}}{c_k(K^\circ)}
  \right).
  \]
  \end{theorem}
\begin{proof}
  Let us fix $k$ and let $E$ be the subspace that achieves
  $c_{k}(K^\circ)$, i.e.~$\diam(\Pi_E(K^\circ)) = c_{k}(K^\circ)$
  and $d_E = \dim E \geq m - k + 1$.  By Corollary~\ref{dualequivalence},
  we have $c_{k}(K^\circ)^{-1} \Pi_E(B_2^m) \subseteq \Pi_E(K).$
  Applying Corollary~\ref{orthproj}, Lemma \ref{lemma:subsetsc}, and
  Corollary~\ref{scaling} in sequence, we get:
  \begin{align*}
  \scz(K,\alpha) \geq \scz(\Pi_E(K), \alpha)  
  &\geq \scz(c_{k}(K^\circ)^{-1} \Pi_E(B^m_2), \alpha)\\
  &= \scz(\Pi(B^{m}_2), \alpha c_{k}(K^\circ) ).
  \end{align*}
  Notice that $\Pi_E(B^m_2)$ is the Euclidean unit ball in the
  subspace $E$, and, therefore: \[\scz(\Pi(B^{m}_2), \alpha
  c_{k}(K^\circ) ) = \scz(B^{d_E}_2, \alpha c_{k}(K^\circ)
  ).\] Finally, by Corollary~\ref{b2bound}, we get the following lower bound,
  as long as $\alpha c_{k}(K^\circ) \le 1/10$:
  \[
  \scz(B^{d_E}_2, \alpha c_{k}(K^\circ)) 
  =
  \Omega{\left(\frac{\sqrt{d_E\log{1/\delta}}}{\alpha c_{k}(K^\circ)
        \epz} \right)} 
  =
  \Omega{\left( \frac{\sqrt{\log{1/\delta}}}{\alpha \epz} \cdot
      \frac{\sqrt{m - k + 1}}{c_{k}(K^\circ)} \right)}.
  \]
  Combining the inequalities completes the proof.
\end{proof}

\section{Optimality of the Gaussian Mechanism}
\label{sect:gauss}

In this section, we present the result that the Gaussian mechanism is
optimal, up to constant factors, when $K\subseteq B_2^m$ is
sufficiently large. More specifically, if the Gaussian mean width of $K$ is
asymptotically maximal, then we can get a tight lower bound on the
sample complexity of the Gaussian mechanism. This is summarized in the
theorem below.

\begin{theorem}\label{thm:gausslb}

  For all $\epz < O(1)$, $2^{-\Omega(n)} \leq \delta \leq 1/n^{1 +
    \Omega{(1})}$,  sufficiently small constant $\alpha$, and
   any symmetric convex body $K \subseteq B_2^m$, if \[ \ell^*(K) =
  \Omega(\sqrt{m}),\] then:
  \[
  \scz(K, \alpha) =  \Theta{\left( \frac{\sqrt{m \log{1/\delta}}}{\alpha \epz} \right)},
  \]
  and $\scz(K,\alpha)$ is achieved, up to constants, by the Gaussian
  mechanism.
\end{theorem}

\junk{This result is not used in the rest of paper, and is included because
it provides a rare example of a general geometric condition under
which the Gaussian mechanism has optimal sample complexity up to
\emph{constant} factors.}

By \autoref{gaussupper} we have an upper bound for the Gaussian
mechanism defined previously. To prove its optimality, we use a
classical result from convex geometry, known as Dvoretzky's criterion,
to show a matching lower bound for the sample complexity. This result
relates the existence of a nearly-spherical section of a given
convex body to the Gaussian mean  norm. It was a key ingredient in
Milman's probabilistic proof of Dvoretzky's theorem: see Matou\v{s}ek's
book~\cite{Matousek-DGeom} for an exposition. 

\junk{In order to relate the
Gaussian mean width to metric entropy, we use Sudakov's minoration
inequality (see~\cite{Ledoux-Talagrand} for a proof).}

\begin{theorem}[\cite{Milman71}; Dvoretzky's Criterion]
  For every symmetric convex body $K \subseteq \R^m$ such that $ B^m_2
  \subseteq K$, and every $\beta < 1$, there exists a constant
  $c(\beta)$ and a subspace $E$ with dimension $\dim{E} \geq c(\beta)
  \ell(K)^2$ for which:
  \[
  (1 - \beta) \frac{\ell(K)}{\sqrt{m}} B_2^m \cap E
  \subseteq K \cap E
  \subseteq (1 + \beta) \frac{\ell^*(K)}{\sqrt{m}} B_2^m \cap E.
  \]
\end{theorem}

\begin{proof}[Proof of \autoref{thm:gausslb}]
  Given the matching upper bound on sample complexity in
  \autoref{gaussupper}, it suffices to show the equivalent lower
  bound, namely that: 
  \[
  \scz(K, \alpha) = 
  \Omega{\left( \frac{\sqrt{m \log{1/\delta}}}{\alpha \epz} \right)}.
  \]
  To this end, we will show that there exists a $k \le (1-c)m + 1$, for
  an absolute constant $c$, so that $c_k(K^\circ) = O(1)$. Then the
  lower bound will follow directly from \autoref{modularlb}.

  We will prove the claim above by applying Dvoretzky's criterion to
  $K^\circ$. By \autoref{polarsubset}, $K \subseteq B_2^m \Rightarrow
  B_2^m \subseteq K^\circ$. We can then apply Dvoretzky's criterion
  with $\beta = 1/2$, ensuring that there exists a subspace $E$ of
  dimension $\dim{E} \ge c(1/2)\ell(K^\circ)^2$ for
  which:
  \[
  K^\circ \cap E \subseteq
  \frac{\ell(K^\circ)}{2\sqrt{m}} B_2^m \cap E.
  \]
  Let us define $k = m - \dim E + 1$; then
  $k \le m - c(1/2)\ell(K^\circ)^2 + 1 = m - c(1/2)\ell^*(K)^2 +
  1$. Since, by assumption $\ell(K^*) = \Omega(m)$, there exists a
  constant $c$ so that $k \le (1-c)m + 1$. Finally, by the definition
  of Gelfand width, $c_k(K^\circ) \le \frac{\ell(K^\circ)}{2\sqrt{m}}
  = O(1)$, as desired. This completes the proof.
\end{proof}



\section{Gaussian Width Lower Bounds in $\ell$-position}
\label{sect:ell}

In Section~\ref{sect:gauss} we showed that the Gaussian Mechanism is
optimal when the Gaussian mean width of $K$ is
asymptotically as large possible. Our goal in this and the following
section is to show general lower bounds on sample complexity in terms
of $\ell^*(K)$. This is motivated by the sample complexity upper
bound in terms of $\ell^*(K)$ provided by the projection mechanism.

It is natural to follow the strategy from Section~\ref{sect:gauss}:
use Dvoretzky's criterion to find a nearly-spherical projection of
$K$ of appropriate radius and dimension. An inspection of the proof of
\ref{thm:gausslb} shows that the sample complexity lower bound we
get this way is $\Omega\left(\frac{\ell^*(K)^2}{\sqrt{m}}\right)$
(ignoring the dependence on $\epz$, $\delta$, and $\alpha$ here, and
in the rest of this informal discussion). Recall that we are aiming
for a lower bound of of $\Omega(\ell^*(K))$, so we are off by a factor
of $\frac{\ell^*(K)}{\sqrt{m}}$. Roughly speaking, the problem is
that Dvoretzky's criterion does too much: it guarantees a spherical
section of $K^\circ$, while we only need a bound on the diameter of
the section. In order to circumvent this difficulty, we use a
different result from asymptotic convex geometry, the low
$M^*$-estimate, which bounds the diameter of a random section of
$K^\circ$, without also producing a large ball contained inside the
section. A technical difficulty is that the resulting upper bound on
the diameter is in terms of the Gaussian mean $K$-norm, rather than the
(reciprocal of the) mean width. When $K$ is in $\ell$-position, this is
not an issue, because results of Pisier, Figiel, and
Tomczak-Jaegermann show that in that case $\ell(K)\ell^*(K) = O(\log
m)$. In this section we assume that $K$ is in $\ell$-position,
and we remove this requirement in the subsequent section.

The main result of this section is summarized below.

\begin{theorem}\label{thm:ellpos}
  For all $\epz = O(1)$, $2^{-\Omega(n)} \leq \delta \leq 1/n^{1 +   \Omega{(1})}$,
  all symmetric convex bodies $K \subseteq \mathbb{R}^m$ in
  $\ell$-position,
  and for $\alpha \le \frac{\ell^*(K)}{C\sqrt{m}\log{2m}}$, where $C$ is
  an absolute constant:
  \[
  \scz{(K,\alpha)} = \Omega{\left( 
      \frac{\sqrt{\log{1/\delta}}}{\alpha \epz}\cdot
      \frac{\ell^*(K)}{\log{2m}}
    \right)}.
  \]
\end{theorem}

The following two theorems are the main technical ingredients we need
in the proof of Theorem~\ref{thm:ellpos}.

\begin{theorem}[\cite{FigielTJ79}, \cite{Pisier-Kconv}; $MM^*$ Bound] \label{thm:mmstar}
  There exists a constant $C$ such that for every symmetric convex body $K \subset \mathbb{R}^m$ in $\ell$-position: \[\ell(K) \cdot \ell^*(K) \leq C \cdot m\log{2m}.\]
\end{theorem}

It is an open problem whether this bound can be improved to
$m\sqrt{\log 2m}$. This would be tight for the cube $Q^m$. This
improvement would lead to a corresponding improvement in our bounds.

\begin{theorem}[\cite{PajorTJ86}; Low $M^*$ estimate]\label{thm:lowmstar}
  There exists a constant $C$ such that for every symmetric convex
  body $K\subset\R^m$ there exists a subspace $E\subseteq \R^m$ with
  $\dim{E} = m - k$ for which: \[\diam{(K \cap E)} \leq C \cdot
  \frac{\ell^*(K)}{\sqrt{k}}.\]
\end{theorem}


Combining Theorems~\ref{thm:mmstar} and~\ref{thm:lowmstar}, we get the
following key lemma.

\begin{lemma}{\label{lm:existence-subsp}}
  There exists a constant $C$ such that for every symmetric convex
  body $K \subset \mathbb{R}^m$ in $\ell$-position, and every $\beta
  \in (0,1-1/m)$, there exists a subspace $E$ of dimension at least
  $\beta m$ satisfying:
  \[
  \diam(K \cap E) \le C\frac{\sqrt{m}\log{2m}}{\sqrt{1-\beta}\cdot\ell(K)}.
  \]
\end{lemma}

\begin{proof}

Let $k = \lfloor (1-\beta) m \rfloor \ge 1$. Using the low $M^*$ estimate
on $K$, there exists a subspace $E$ with $\dim{E} = m - k = \lceil
\beta m \rceil$ for which: 
\[
\diam{(K \cap E)} \leq C_1 \cdot \frac{\ell^*(K)}{\sqrt{k}}.
\]
By the $MM^*$ upper bound, since $K$ is in $\ell$-position, we have
that:
\[
\ell^*(K) \leq C_2 \cdot \frac{m \log{2m}}{\ell(K)},
\]
and, combining the two inequalities, we get that:
\[
\diam(K\cap E) \le C_1C_2 \frac{m\log{2m}}{\sqrt{k}\cdot\ell(K)} \le
C\frac{\sqrt{m}\log m}{\sqrt{1-\beta} \cdot \ell(K)},
\]
for an appropriate constant $C$. This completes the proof.

\end{proof}

The proof of the desired lower bound now follows easily from this lemma.

\begin{proof}[Proof of Theorem~\ref{thm:ellpos}]

By Theorem~\ref{modularlb}, it suffices to show that
\begin{equation}\label{eq:ellpos-lb}
\max_{k = 1}^m \frac{\sqrt{m - k + 1}}{c_k(K^\circ)} 
= \Omega{\left(\frac{\ell^*(K)}{\log{2m}} \right)}.
\end{equation}
Indeed, if $k^*$ is the value of $k$ for which the maximum on the
left hand side is achieved, then $\frac{\sqrt{m - k^* +
    1}}{c_{k^*}(K^\circ)}$ is a lower bound on the sample complexity
for all $\alpha \le 1/(10c_{k^*}(K^\circ))$, and by
\eqref{eq:ellpos-lb}:
\[
\frac{1}{10c_{k^*}(K^\circ)} 
= \Omega\left(\frac{\ell^*(K)}{\sqrt{m-k+1}\cdot \log{2m}}\right)
= \Omega\left(\frac{\ell^*(K)}{\sqrt{m}\cdot \log{2m}}\right).
\]
In the rest of the proof, we establish \eqref{eq:ellpos-lb}.

Since $K$ (and thus also $K^\circ$) are in $\ell$-position by assumption, from
Lemma \ref{lm:existence-subsp} applied to $K^\circ$ we have that
there exists a subspace $E$ such that $\dim{E} \ge m/2$ and:
\[
\diam(K^\circ \cap E)
= O\left(\frac{\sqrt{m}\log 2m}{\ell(K^\circ)}\right)
= O\left(\frac{\sqrt{m}\log 2m}{\ell^*(K)}\right).
\]
Setting $k_E = m - \dim{E} + 1 \le m/2 + 1$, and because $c_{k_E}(K^\circ)
\le \diam(K^\circ \cap E)$ by definition, we get that:
\[
\max_{k = 1}^m \frac{\sqrt{m - k + 1}}{c_k(K^\circ)} 
\ge \frac{\sqrt{\dim{E}}}{\diam(K^\circ \cap E)} = \Omega{\left(\frac{\ell^*(K)}{\log{2m}} \right)},
\]
as desired.

\end{proof}

\junk{
This lower bound, along with \autoref{thm:general-proj}, give us the optimality result below for the projection mechanism.

\begin{corr}
If $\epz = O(1)$, $2^{-\Omega(n)} \leq \delta \leq 1/n^{1 + \Omega{(1})}$ and $\alpha \leq 1/10$, the Projection mechanism is optimal for constant $\alpha$ up to a $\log{2m}$ factor.
\end{corr}}
\section{Gaussian Width Lower Bounds for Arbitrary Bodies}
\label{sect:arbirary}

In this section, we remove the assumption that $K$ is in
$\ell$-position from the previous section. Instead, we use a recursive
charging argument in order to reduce to the $\ell$-position case. The
resulting guarantee is worse than the one we proved for bodies in
$\ell$-position by a logarithmic factor.

The main lower bound result of this section is the following theorem.
\begin{theorem}\label{thm:gen-lb}
  For all $\epz = O(1)$, $2^{-\Omega(n)} \leq \delta \leq 1/n^{1
    +\Omega(1)}$, any symmetric convex body $K\subset \R^m$, and
  any $\alpha \leq \frac{\ell^*(K)}{C\sqrt{m}(\log 2m)^2}$, where $C$
  is an absolute constant: 
  \[
  \scz(K, \alpha) =
  \Omega\left(\frac{\sigma(\epz,\delta)\ell^*(K)}{(\log 2m)^2\alpha}\right).
  \]
\end{theorem}

The lower bound follows from the geometric lemma below, which is
interesting in its own right.

\begin{lemma}\label{lm:volnum}
  There exists a constant $C$ such that, for any symmetric convex body
  $K \subset \R^m$, 
  \begin{equation}\label{eq:volnum}
    \ell^*(K) \le C(\log 2m) 
    \left(
      \sum_{i = 1}^m{\frac{1}{\sqrt{i} \cdot c_{m-i+1}(K^\circ)}}
    \right).
\end{equation}
\end{lemma}

Lemma \ref{lm:volnum} is closely related to the volume number theorem of
Milman and Pisier~\cite{MilmanPisier87}, which states that the
inequality \eqref{eq:volnum} holds with $\frac{1}{c_{m-i+1}(K^\circ)}$
replaced by the volume number $v_{i}(K)$, defined as:
\[
v_i(K) = \sup_{E: \dim{E} = i}\frac{\vol(\Pi_E(K))^{1/i}}{(\vol(\Pi_E(B_2^m))^{1/i}}, 
\]
where the supremum is over subspaces $E$ of $\R^m$. Inequality \eqref{eq:volnum} is stronger
than the volume number theorem, because $\frac{1}{c_{m-i+1}(K^\circ) }
\le v_{i}(K)$. Indeed, setting $r = \frac{1}{c_{m-i+1}(K^\circ) }$, by
\autoref{dualequivalence} and the definition of Gelfand width we have
that there exists a subspace $E$ of dimension $i$ such that $r
\Pi_E(B_2^m)\subseteq \Pi_E K$. Therefore, $\vol(\Pi_E(K)) \ge
r^{i}\vol(\Pi_E(B_2^m))$, which implies the desired inequality.

Even though the volume number theorem is weaker than
\eqref{eq:volnum}, the proof given by Pisier in his
book~\cite{Pisier-book}, with minor modifications, appears to yield
the stronger inequality we need. Rather than repeat this argument, we
give a self-contained and slightly different proof below. Our proof
only uses the low $M^*$ estimate, the $MM^*$ estimate, and elementary
linear algebra, while Pisier's proof uses Milman's quotient of
subspace theorem. Moreover, if the $MM^*$ estimate can be improved to
$O(m\sqrt{\log 2m})$, our argument would imply a corresponding
improvement of the logarithmic factor in \eqref{eq:volnum} from $\log
2m$ to $\sqrt{\log 2m}$. This does not appear to be the case in
Pisier's proof, where the logarithmic factor comes from the (tight)
upper estimate of the $K$-convexity constant of $m$-dimensional Banach
spaces.

To prove Lemma \ref{lm:volnum} we will first establish an auxiliary
lemma.

\begin{lemma}\label{lm:volnum-proj}
  There exists a constant $C$ such that, for any symmetric convex body
  $K \subset \R^m$ and $m \ge 4$, there exists a subspace $E \subseteq \R^m$ such
  that $\dim{E} \ge \lfloor m/4 \rfloor$ and: 
  \[
  \ell^*(\Pi_E(K)) \le C(\log 2m) \frac{\sqrt{m}}{c_{2\lfloor m/4
      \rfloor + 1}(K^\circ)}.
  \]
\end{lemma}

\begin{proof}
  Let $T$ be a linear operator such that $T(K)$ is in
  $\ell$-position. By the rotational invariance of Gaussians,
  $\ell^*(UT(K)) = \ell^*(T(K))$ and $\ell(UT(K)) = \ell(T(K))$ for
  any orthogonal transformation $U$; so, 
  we may assume that $T$ is self-adjoint and positive definite.

  Let $\lambda_1 \ge \ldots \ge \lambda_m > 0$ be the eigenvalues of
  $T$. Let us set $k = \lfloor m/4 \rfloor$ and define $E$ to be the
  subspace spanned by the eigenvectors of $T$ corresponding to
  $\lambda_1, \ldots, \lambda_{k}$. Observe that
  $\Pi_E$ and $T$ commute because they are simultaneously
  diagonalized by the eigenvectors of $T$.

  Using Lemma \ref{lm:comp-width}, we calculate:
  \[
  \ell^*(\Pi_E(K)) = \ell^*(\Pi_ET^{-1}T(K))
  \le \|\Pi_ET^{-1}\|\ell^*(T(K)).
  \]
  By the definition of $E$, the singular values of $\Pi_E T^{-1}$ are
  $\lambda_{k}^{-1} \ge \ldots \ge \lambda_1^{-1}$, and therefore
  the operator norm of $\Pi_E T^{-1}$ is $\|\Pi_E T^{-1}\| =
  \lambda_{k}^{-1}$.  Thus we have:
  \begin{equation}
    \label{eq:width-ub}
    \ell^*(\Pi_E(K)) \le \frac{\ell^*(T(K))}{\lambda_k}.
  \end{equation}

  Since $T(K)$, and thus $T(K)^\circ$, is in $\ell$-position,
  we can apply Lemma \ref{lm:existence-subsp} to $T(K)^\circ$ and $\beta
  = \frac{m-k}{m} \ge \frac14$ and get that there exists a subspace $F$ of
  dimension at least $m-k$ so that:
  \[
  \ell^*(T(K)) \le C_1(\log{2m})\frac{\sqrt{m}}{\diam(T(K)^\circ \cap   F)},
  \]
  where $C_1$ is an absolute constant.

  Let $G = E^\perp$ be the orthogonal complement of the subspace $E$,
  i.e.~the space spanned by the eigenvectors of $T$ corresponding to
  the eigenvalues $\lambda_{k+1}, \ldots, \lambda_m$. The restriction $T|_G:G \to \R^m$
  of $T$ to $G$ has
  eigenvalues $\lambda_{k + 1} \ge \ldots \ge \lambda_m$. Therefore
  $\|T|_G\| \le \lambda_{k + 1} \le \lambda_{k}$. Because $G$ is spanned by
  eigenvectors of $T$, it is invariant under action by $T$, i.e.~$T(G) =
  G$, and, equivalently, $T^{-1}(G) = G$. Therefore $T^{-1}(K^\circ)
  \cap G = T^{-1}(K^\circ \cap G)$ and $K^\circ \cap G =
  (T|_G) T^{-1}(K^\circ \cap G)$. It follows that: 
  \begin{align*}
  \diam(K^\circ \cap F \cap G) 
  &= \diam(  (T|_G) T^{-1}(K^\circ \cap G) \cap F)\\
  &\le \|T|_G\| \diam(T^{-1}(K^\circ \cap G) \cap F)\\
  &\le \lambda_k \diam(T^{-1}(K^\circ) \cap F)\\
  &= \lambda_k \diam(T(K)^\circ \cap F).
  \end{align*}
  Setting $H = F \cap G$ and combining the inequalities, we get that:
  \[
  \ell^*(T(K)) \le 
  C_1(\log{2m}) \frac{\lambda_k\sqrt{m}}{\diam(K^\circ  \cap H)}
  \le  C_1(\log{2m}) \frac{\lambda_k\sqrt{m}}{c_{2k+1}(K^\circ)},
  \]
  where the final inequality follows from the definition of Gelfand
  width, since the co-dimension of $H$ is at most $2k$. Sibstituting
  this inequality into the right hand side of \eqref{eq:width-ub} finishes
  the proof of the lemma.
\end{proof}

We are now ready to prove Lemma~\ref{lm:volnum}.

\begin{proof}[Proof of Lemma~\ref{lm:volnum}]
  Let us first establish the lemma for $m \le 16$. Notice that:
  \[
  c_m(K^\circ) = \min_{\theta: \|\theta\|_2 = 1}{\diam(K^\circ \cap
    \{t\theta: t\in \R\})}.
  \] 
  But $\diam(K^\circ \cap \{t\theta: t\in \R\}) =
  \frac{1}{\|\theta\|_{K^\circ}}$ for any unit vector $\theta$, by the
  definition of Minkowski norm. Therefore, $1/c_m(K^\circ) =
  \max\{\|\theta\|_{K^\circ}: \|\theta\|_2 = 1\}$. It follows that,
  for a standard Gaussian $g$ in $\R^m$:
  \begin{equation}\label{eq:small-m}
  \ell^*(K) = \E \|g\|_{K^\circ} 
  \le \frac{\E\|g\|_2}{c_m(K^\circ)} 
  \le \frac{(\E\|g\|^2_2)^{1/2}}{c_m(K^\circ)}
  = \frac{\sqrt{m}}{c_m(K^\circ)}  \le \frac{4}{c_m(K^\circ)} .
  \end{equation}
  This establishes the lemma for $m \le 16$.

  For $m > 16$, we will use an induction argument. We will
  strengthen the induction hypothesis to:
  \begin{equation}\label{eq:induct-hyp}
  \ell^*(K) \le C(\log 2k) 
  \left(
    \sum_{i = 1}^{\lceil k/2 \rceil}{\frac{1}{\sqrt{i} \cdot c_{k-i+1}(K^\circ)}}
  \right).
  \end{equation}
 Assume the inequality
   holds for a sufficiently large absolute
  constant $C$ and all symmetric convex bodies $K \subseteq \R^{k}$,
  in any dimension $k < m$. We will show that the inequality then
  holds in dimension $k = m$ as well. The inequality
  \eqref{eq:small-m},  provides the base case  for the induction ($k \le 16$).

  We proceed with the inductive step. By
  Lemma \ref{lm:volnum-proj} there exists a subspace $E$ of dimension at
  least $\lfloor m/4 \rfloor$ such that:
  \[
  \ell^*(\Pi_E(K)) \le C_1(\log 2m) \frac{\sqrt{m}}{c_{2\lfloor m/4
      \rfloor + 1}(K^\circ)}.
  \]
  Observe that, since $2\lfloor m/4 \rfloor \le \lfloor m/2 \rfloor$,
  and $c_j(K^\circ)$ is monotone non-increasing in $j$, we have that
  $c_{2\lfloor m/4   \rfloor + 1}(K^\circ) \ge c_{m-i+1}(K^\circ)$ for
  any $i \le \lceil m/2 \rceil$. Moreover, 
  \[
  \sum_{i = \lceil 3m/8 \rceil + 1}^{\lceil m/2 \rceil}{\frac{1}{\sqrt{i}}} \ge \frac{1}{C_2} \sqrt{m},
  \]
  for all $m > 16$ and an absolute constant $C_2$. These observations
  together imply that 
  \begin{equation}\label{eq:volnum-proj}
  \ell^*(\Pi_E(K)) \le C_1(\log 2m)\frac{\sqrt{m}}{c_{2\lfloor m/4  \rfloor +   1}(K^\circ)} 
  \le C_1C_2(\log 2m) \sum_{i = \lceil 3m/8 \rceil + 1}^{\lceil m/2  \rceil}{\frac{1}{\sqrt{i}\cdot c_{m-i+1}(K^\circ)}}. 
  \end{equation}

  For a standard Gaussian $g \sim N(0,I_m)$, using the triangle
  inequality and the fact that $I_m = \Pi_E + \Pi_{E^\perp}$, we have:
  \begin{align*}
    \ell^*(K) &= \E \max_{x \in K} |\langle x, g\rangle|\\
    &= \E \max_{x \in K} \left|\langle \Pi_{E} x, g\rangle + \Pi_{E^\perp} x, g\rangle\right|\\ 
    &\le  \E \max_{x \in K} |\langle \Pi_{E} x,   g\rangle| + \E \max_{x \in K} |\langle \Pi_{E^\perp} x,   g\rangle|\\
    &= \ell^*(\Pi_{E}K) + \ell^*(\Pi_{E^\perp}K).
  \end{align*}
  Then, the inductive hypothesis follows by \eqref{eq:volnum-proj} and by the inductive hypothesis \eqref{eq:induct-hyp} applied
  to $\ell^*(\Pi_{E^\perp})$. This finishes the proof of the lemma.
  \junk{by \eqref{eq:volnum-proj}, and by the inductive hypothesis \eqref{eq:induct-hyp} applied
  to $\ell^*(\Pi_{E^\perp})$, we have:
  \[
  \ell^*(K)  \le 
  C(\log 2\lceil 3m/4\rceil) 
  \left(
    \sum_{i = 1}^{\lceil 3m/8 \rceil}{\frac{1}{\sqrt{i} \cdot c_{k-i+1}(K^\circ)}}
  \right)
  +
  C_1C_2(\log 2m) 
  \left(
  \sum_{i = \lceil 3m/8 \rceil + 1}^{\lceil m/2 \rceil}{\frac{1}{\sqrt{i}\cdot c_{m-i+1}(K^\circ)}}
  \right)
  \le 
  C(\log 2m) 
  \left(
    \sum_{i = 1}^{\lceil m/2 \rceil}{\frac{1}{\sqrt{i} \cdot c_{k-i+1}(K^\circ)}}
  \right),
  \]
  where we assumed that $C \ge C_1C_2$. This finishes the inductive
  step, and the proof of the lemma.}
\end{proof}

\begin{proof}[Proof of Theorem~\ref{thm:gen-lb}]
  As in the proof of \autoref{thm:ellpos}, it is sufficient to prove
  that:
  \begin{equation}\label{eq:gen-lb}
  \max_{k = 1}^m \frac{\sqrt{m - k + 1}}{c_k(K^\circ)} 
  = \Omega{\left(\frac{\ell^*(K)}{(\log{2m})^2} \right)}.
  \end{equation}
  But this inequality follows easily from Lemma \ref{lm:volnum} and the
  trivial case of H\"older's inequality:
  \begin{align*}
    \ell^*(K) 
    \le C(\log 2m) 
    \left(
      \sum_{i = 1}^m{\frac{1}{\sqrt{i} \cdot c_{m-i+1}(K^\circ)}}
    \right)
    &\le C(\log 2m)
    \left(
      \sum_{i = 1}^m{\frac{1}{i}}
    \right)
    \cdot
    \left(\max_{i = 1}^m \frac{\sqrt{i}}{c_{m-i+1}(K^\circ)} \right)\\
    &= O((\log{2m})^2)
    \cdot
    \left(\max_{i = 1}^m \frac{\sqrt{i}}{c_{m-i+1}(K^\circ)} \right).
  \end{align*}
  Then, the proof of the theorem follows from \eqref{eq:gen-lb}
  analogously to the proof of \autoref{thm:ellpos}.
\end{proof}

We now have everything in place to prove our main result for the mean
point problem. 

\begin{proof}[Proof of \autoref{thm:meanpt-main}]
  The upper bounds on sample complexity follow from
  \autoref{thm:gauss-mech}, \autoref{gaussupper}, 
  \autoref{thm:general-proj}, and \autoref{projupper}. The lower
  bounds follow from \autoref{thm:gen-lb}. The statement after
  ``moreover'' follows from \autoref{gaussupper} and
  \autoref{thm:gausslb}. 
\end{proof}

\junk{
\begin{theorem}\label{thm:meanpt-main}
  For any symmetric convex body $K \subseteq B_2^m$, and for all $\epz
  = O(1)$, $2^{-\Omega(n)} \leq \delta \leq 1/n^{1 +\Omega(1)}$, and
  $\alpha \leq 1/10$, the $(\epz, \delta)$-differentially private
  Projection Mechanism $\alg_{K}$ of Theorem~\ref{thm:general-proj}
  satisfies 
  \[
  \samp(K, \alg_{K}, \alpha) = O\left(\frac{(\log m)^2}{\alpha}\right) \scz(K,
  \alpha). 
  \]
  Moreover, $\scz(K, \alpha)$ satisfies
  \begin{align*}
    \scz(K,\alpha) &= O\left(\frac{\sigma(\epz,
        \delta)\ell^*(K)}{\alpha^2}\right);\\
    \scz(K,\alpha) &= \Omega\left(\frac{\sigma(\epz,
        \delta)\ell^*(K)}{\alpha (\log m)^2}\right).
  \end{align*}
\end{theorem}
\begin{proof}
  Follows immediately from
  Theorems~\ref{thm:general-proj}~and~\ref{thm:gen-lb}, and
  Corollary~\ref{projupper}.
\end{proof}

Theorem~\ref{thm:main} follows immediately from
Theorem~\ref{thm:meanpt-main}, Lemma~\ref{lm:mean-pt-equiv}, and the
observation that for any workload $\quer$ of $m$ queries with
sensitivity polytope $K$, and any sufficiently small $\epz$ and
$\delta$, $\scz(\quer,\alpha) = \Omega\left(
  \frac{\diam(K)}{\sqrt{m}\alpha}\right)$. 
}

\section{From Mean Point to Query Release}

All the bounds we proved so far were for the mean point problem. In
this section we show reductions between this problem, and the query
release problem, which allow us to translate our lower bounds to the
query release setting and prove Theorem~\ref{thm:main}. We will show that
the problem of approximating $\quer(\db)$ for a query workload $\quer$
under differential privacy is nearly equivalent to approximating the
mean point problem with universe $K' = \frac{1}{\sqrt{m}}K$, where $K$
is the sensitivity polytope of $\quer$. The main technical lemma follows.

\begin{lemma}\label{lm:mean-pt-equiv}
  Let $\quer$ be a workload of $m$ linear queries over the universe
  $\uni$ with sensitivity polytope $K$. Define $K' =
  \frac{1}{\sqrt{m}}K$. Then, we have the inequalities:
  \begin{align}
    \scz(\quer, \alpha) &\leq \scz(K', \alpha);\label{eq:mean-pt-equiv-1}\\
    \scz(K', \alpha) &\leq \max\left\{\samp_{2\epz, 2\delta}(\quer, \alpha/4), \frac{16\diam(K')}{\alpha^2}\right\}.\label{eq:mean-pt-equiv-2}
  \end{align}
  Moreover, we can use an $(\epz, \delta)$-differentially private
  algorithm $\alg'$ as a black box to get an $(\epz,
  \delta)$-differentially private algorithm $\alg$ such that
  $\samp(\quer, \alg, \alpha) = \samp(K', \alg', \alpha)$. $\alg$
  makes a single call to $\alg'$, and performs additional computation
  of worst-case complexity $O(mn)$, where $n$ is the size of the database.
\end{lemma}

\begin{proof}
  Note that inequality~\eqref{eq:mean-pt-equiv-1} is implied by the
  statement after ``moreover''. We prove this claim first. The
  algorithm $\alg$ uses $\db$ to form a database $\db'$ of elements
  from $K'$ which contains, for each $\row \in \db$, a copy of
  $\frac{1}{\sqrt{m}} \quer(\{\row\})$. This transformation clearly
  takes time $O(mn)$. Then $\alg$ simply outputs $\sqrt{m}
  \alg'(\db')$. Because any two neighboring databases $\db_1$, $\db_2$
  drawn from $\uni$ are transformed into neighboring databases
  $\db'_1$, $\db'_2$ of points from $K'$, and $\alg'$ was assumed to
  be $(\epz, \delta)$-differentially private, we have that $\alg$ is
  $(\epz, \delta)$-differentially private as well. Also, since
  $\avgdb' = \frac{1}{\sqrt{m}} \quer(\db)$, $\err(\quer,\alg,n) =
  \err(K',\alg,n)$ by definition, which implies that $\samp(\quer,
  \alg, \alpha) = \samp(K', \alg', \alpha)$.

  The second inequality~\eqref{eq:mean-pt-equiv-2} is more
  challenging. We will show that for any $(\epz,
  \delta)$-differentially private algorithm $\alg$ there exists an
  $(2\epz, 2\delta)$-differentially private algorithm $\alg'$ such
  that:
  \[
  \err(K', \alg', n) \le 2\err(\quer, \alg, n) + \frac{2\diam(K')}{\sqrt{n}}. 
  \]
  A simple calculation then shows that this implies the desired
  inequality. In constructing $\alg'$ we will run $\alg$ on a database
  formed by sampling from the vertices of $K'$ (which correspond to
  universe elements) so that the true query answers are preserved in
  expectation. The analysis uses symmetrization.
  
  Let $\alg'$ be given the input $\db' = \{x_1, \ldots, x_n\}$, which
  is a multiset of points from $K'$.  $\alg'$ will randomly construct
  two databases $\db_+$ and $\db_-$ (i.e.~multisets of elements from
  $\uni$), and output $\frac{1}{\sqrt{m}}(\alg(\db_+) - \alg(\db_-))$.

  Next we describe how $\db_+$ and $\db_-$ are sampled. Observe that,
  for each $i$, $x_i$ is a convex combination of vertices of
  $K'$. Therefore, by Caratheodory's theorem, there exist universe
  elements $\row_{i,1}, \ldots, \row_{i, k_i} \in \uni$, where $k_i
  \le m+1$, such that\junk{ the vectors $\quer(\row_{i,1}), \ldots,
    \quer(\row_{i,k_i})$ are affinely independent, and} $x_i =
  \frac{1}{\sqrt{m}}\sum_{j = 1}^{k_i}{\alpha_{i,j}\quer(\row_{i,j})}$
  for some $\alpha_{i,1}, \ldots, \alpha_{i,k_i} \in [-1,1]$
  satisfying $\sum_{j=1}^{k_i}{|\alpha_i|} = 1$. We would like to fix
  a unique way to pick the $\row_{i,j}$ and $\alpha_{i,j}$ for each
  $x_i$, and indeed for any point in $K'$. To this end, fix an
  arbitrary order on $\uni$, and for each $x_i$ choose a minimal
  sequence $\row_{i,1}, \ldots, \row_{i,k_i}$ that satisfies the
  conditions above and which comes earliest in the lexicographic order
  induced by the order on $\uni$.  Once we have chosen the
  $\row_{i,j}$ and $\alpha_{i,j}$, we construct $\db_+$ and $\db_-$
  using the following sampling procedure: for each $i \in [n]$, we
  independently sample $j_i$ from $[k_i]$, s.t.~$\Pr[j_i
  = j] = |\alpha_{i,j}|$; we add $\row_{i,j_i}$ to $\db_+$ if
  $\alpha_{i,j_i} \ge 0$, and to $\db_-$ otherwise. Then, as mentioned
  above, $\alg'$ outputs $\frac{1}{\sqrt{m}}(\alg(\db_+) -
  \alg(\db_-))$.

  First, we show that $\alg'$ is $(2\epz, 2\delta)$-differentially
  private. Let $\db_1' = \{x_1, \ldots, x_n\}$ and $\db_2' = \{x_1,
  \ldots, x_n, x_{n+1}\}$ be two neighboring databases and let
  $\{\row_{i,j}\}$ and $\{\alpha_{i,j}\}$ be defined as above, where
  $1 \le i \le n+1$ and for each $i$, $1 \le j \le k_i$. Let
  $\db_+^{(1)}, \db_-^{(1)}$ be the random databases sampled by
  $\alg'$ on input $\db'_1$, and let $\db_+^{(2)}, \db_-^{(2)}$ be the
  random databases sampled on input $\db'_2$. Let us couple
  $\db_+^{(1)}$ and $\db_+^{(2)}$ so that they have symmetric
  difference at most $1$ and $\db_+^{(2)}$ is a superset of
  $\db_+^{(1)}$. We can achieve this by sampling $j_i \in [k_i]$ as described above for $1 \le i \le n$ and adding
  $\row_{i,j_i}$ to both $\db_+^{(1)}$ and $\db_+^{(2)}$ if
  $\alpha_{i,j_i} \ge 0$ (or to neither otherwise) and sampling
  $j_{n+1} \in [k_{n+1}]$ and adding $\row_{n+1, j_{n+1}}$
  to $\db_+^{(2)}$ if $\alpha_{n+1, j_{n+1}} \ge 0$. Because, with this
  coupling, $\db_+^{(1)}$ and $\db_+^{(2)}$ are always neighboring,
  by the $(\epz,\delta)$-differential privacy guarantee for $\alg$
  we have that, for any measurable $S$ in the range of $\alg$:
  \begin{align*}
  \Pr[\alg(\db_+^{(1)}) \in S] &\le e^{\epz}  \Pr[\alg(\db_+^{(2)})
  \in S]  + \delta,\\
  \Pr[\alg(\db_+^{(1)}) \in S] &\ge e^{-\epz}  \Pr[\alg(\db_+^{(2)})
  \in S]  -\delta,
  \end{align*}
 where the probabilities are taken over the random choices of
 $\db_+^{(1)}$ and $\db_+^{(2)}$, and over the randomness of
 $\alg$. An analogous argument holds for $\db_-^{(1)}$ and
 $\db_-^{(2)}$. Therefore, each of the two calls of $\alg$ made by
 $\alg'$, composed with the sampling procedure that produces the input
 to $\alg$, is $(\epz, \delta)$-differentially private \emph{with
   respect to the input to $\alg'$}.  The privacy claim follows by
 composition (Lemma~\ref{lm:composition}).

 Finally we analyze the error of $\alg'$. By Minkowski's inequality, we have:
 \begin{align*}
   (\E \|\alg'(\db') - \avgdb'\|_2^2)^{1/2}
   &\le \Big(\E \Big\|\alg(\db_+) - \frac{1}{\sqrt{m}}\quer(\db_+)\Big\|_2^2\Big)^{1/2} 
    + \Big(\E\Big\|\alg(\db_-) - \frac{1}{\sqrt{m}}\quer(\db_-)\Big\|_2^2\Big)^{1/2}\\
   &\hspace{4.5em}+ \Big(\E\Big\|\frac{1}{\sqrt{m}}(\quer(\db_+) -
   \quer(\db_-)) - \avgdb'\Big\|_2^2\Big)^{1/2}\\
   &\le 2\err(\quer, \alg, n) + \Big(\E\Big\|\frac{1}{\sqrt{m}}(\quer(\db_+) -
   \quer(\db_-)) - \avgdb'\Big\|_2^2\Big)^{1/2}.
 \end{align*}
 We proceed to bound the second term on the right hand side. Observe
 that, by the way we sample $\db_+$ and $\db_-$, linearity of
 expectation, and the linearity of $\quer(\cdot)$, $\avgdb' = \E
 \frac{1}{\sqrt{m}}(\quer(\db_+) - \quer(\db_-))$. Let, for each $i
 \in [n]$, $j_i \in [k_i]$ be sampled as above, and let $y_i =
 \frac{1}{\sqrt{m}}\sign(\alpha_{i,j})\quer(e_{i,j_i})$. Then
 $\frac{1}{\sqrt{m}}(\quer(\db_+) - \quer(\db_-)) = \frac1n
 \sum_{i=1}^n{y_i}$, so $\E \frac1n \sum_{i=1}^n{y_i} = \avgdb'$. The
 vectors $y_1, \ldots, y_n$ are independent and all belong to
 $K'\subseteq R\cdot B_2^m$, where $R = \diam(K')$. Let $\xi_1,
 \ldots, \xi_n$ be Rademacher random variables, independent from
 everything else; by
 Lemma \ref{lm:sym} and the parallelogram identity:
 \begin{align*}
   \E\Big\|\frac{1}{\sqrt{m}}(\quer(\db_+) -  \quer(\db_-)) -
   \avgdb'\Big\|_2^2
   &= \E\Big\|\frac1n \sum_{i=1}^n{y_i} - \E\frac1n \sum_{i=1}^n{y_i}\Big\|_2^2\\
   &\le \frac{4}{n^2} \E\Big\|\sum_{i=1}^n{\xi_iy_i}\Big\|_2^2
   = \frac{4}{n^2}\sum_{i=1}^n{\|y_i\|_2^2} \le \frac{4R^2}{n}.
 \end{align*}
 With this, we have established the desired bound on
 $\err(K',\alg',n)$, and, therefore, the lemma.
\end{proof}

We will also use a simple lemma that relates the sample complexity at
an error level $\alpha$ to the sample complexity at a lower error
level $\alpha' < \alpha$. The proof is a padding argument and
can be found in~\cite{SteinkeU15}. 

\begin{lemma}\label{lm:padding}
  For any workload $\quer$, any $0 < \alpha' < \alpha < 1$, and any
  privacy parameters $\epz, \delta$, we have
  \[
  \scz(\quer,\alpha') = \Omega\left(\frac{\alpha}{C\alpha'}\right)\cdot \scz(\quer,\alpha),
  \]
  for an absolute constant $C$.
\end{lemma}

We are now ready to finish the proof of our main result for query
release.

\begin{proof}[Proof of \autoref{thm:main}]
  The upper bounds on sample complexity follow from
  the upper bounds in \autoref{thm:meanpt-main} together with
  Lemma \ref{lm:mean-pt-equiv}. 

  Denote $K' = \frac{1}{\sqrt{m}} K$, and let $\alpha_0 =
  \frac{\ell^*(K')}{C\sqrt{m}(\log 2m)^2} =
  \frac{\ell^*(K)}{Cm(\log 2m)^2}$ be the smallest error parameter for
  which \autoref{thm:gen-lb} holds. Then, by
  \autoref{thm:gen-lb} and Lemma \ref{lm:mean-pt-equiv}:
  \[
  \max\left\{\scz(\quer, \alpha_0), \frac{\diam(K)}{\sqrt{m}\alpha_0^2}
  \right\}  
  = \Omega\left(\frac{\sigma(\epz,\delta)\ell^*(K)}{\sqrt{m}(\log 2m)^2\alpha_0} \right).
  \]
  It is easy to show that $\scz(\quer, \alpha_0) =
  \Omega(\diam(K)/(\alpha_0\sqrt{m}))$ for all sufficiently small
  $\epz$  and $\delta$.  Therefore, we have: 
  \[
  \scz(\quer, \alpha_0) = 
  \Omega\left(\frac{\sigma(\epz,\delta)\ell^*(K)}{\sqrt{m}(\log 2m)^2} \right).
  \]
  By Lemma \ref{lm:padding}, we get that for any $\alpha \le \alpha_0$
  the sample complexity is at least:
  \[
  \scz(\quer, \alpha) = 
  \Omega\left(\frac{\sigma(\epz,\delta)\ell^*(K)\alpha_0}{\sqrt{m}(\log  2m)^2\alpha} \right)
  =     \Omega\left(\frac{\sigma(\epz,\delta)\ell^*(K)^2}{m^{3/2}(\log  2m)^4\alpha} \right).
  \]

  An analogous proof, with $\alpha_0 = 1/C$ set to the smallest error
  parameter for which \autoref{thm:gausslb} holds, establishes the
  statement after ``moreover''.
\end{proof}

\bibliographystyle{alpha}
\bibliography{Privacy,sasho-papers}

\end{document}